\documentclass[journal]{IEEEtran}
\usepackage{algorithm,amsbsy,amsmath,amssymb,epsfig,bbm,mathrsfs,multirow,amsthm}
\usepackage{array,multirow,graphicx}
\usepackage{graphicx}
\usepackage{graphics}
\usepackage{epstopdf}
\usepackage{color}
\usepackage{bbm}
\usepackage{comment}
\usepackage{cite}
\usepackage{algpseudocode}
\usepackage{mathtools,xparse}
\usepackage{subcaption}
\DeclarePairedDelimiter{\norm}{\lVert}{\rVert}

\newcommand{\beq}{\begin{equation}}
\newcommand{\eeq}{\end{equation}}
\newcommand{\bitm}{\begin{itemize}}
\newcommand{\ba}{\begin{array}}
\newcommand{\ea}{\end{array}}
\newcommand{\eitm}{\end{itemize}}
\newcommand{\beqn}{\begin{eqnarray}}
\newcommand{\eeqn}{\end{eqnarray}}
\newcommand{\beqno}{\begin{eqnarray*}}
\newcommand{\eeqno}{\end{eqnarray*}}
\newcommand{\bma}{\begin{displaymath}}
\newcommand{\ema}{\end{displaymath}}
\newcommand{\bnu}{\begin{enumerate}}
\newcommand{\enu}{\end{enumerate}}
\newcommand{\bce}{\begin{center}}
\newcommand{\ece}{\end{center}}
\newcommand{\btb}{\begin{tabular}}
\newcommand{\etb}{\end{tabular}}

\newtheorem{theorem}{\textbf{\textsc{Theorem}}}

\begin{document}
\title{\huge Dynamic Network Service Selection in IRS-Assisted Wireless Networks: A Game Theory Approach}

\author{ 
\IEEEauthorblockN{Nguyen Cong Luong,
Nguyen Thi Thanh Van,
Feng Shaohan,
Huy T. Nguyen,
Dusit Niyato,~\IEEEmembership{Fellow,~IEEE}
and
Dong In Kim,~\IEEEmembership{Fellow,~IEEE}}

\thanks{N. C. Luong is with the Faculty of Computer Science, PHENIKAA University,
Hanoi 12116, Vietnam. E-mail:luong.nguyencong@phenikaa-uni.edu.vn.}
\thanks{N. T. T. Van is with the Hanoi Vocational College of High Technology, Vietnam. Email: vanntt@haui.edu.vn.}
\thanks{S. Feng, H. T. Nguyen and D. Niyato are with the School of Computer Science and Engineering, Nanyang Technological University, Singapore. Emails: feng0089@e.ntu.edu.sg,
huyt.nguyen@ntu.edu.sg, dniyato@ntu.edu.sg.}
\thanks{D.~I.~Kim is with School of Information and Communication Engineering, Sungkyunkwan University, Korea. Email: dikim@skku.ac.kr.}
\vspace{-0.7cm}
}

\maketitle
\begin{abstract}
In this letter, we investigate the dynamic network service provider (SP) and service selection in an intelligent reflecting surface (IRS)-assisted wireless network. In the network, mobile users select different network resources, i.e., transmit power and IRS resources, provided by different SPs. To analyze the SP and network service selection of the users, we formulate an evolutionary game. In the game, the users (players) adjust their selections of the SPs and services based on their utilities. We model the SP and service adaptation of the users by the replicator dynamics and analyze the equilibrium of the evolutionary game. Extensive simulations are provided to demonstrate consistency with the analytical results and the effectiveness of the proposed game approach. 
\end{abstract}

\section{Introduction}
To improve the spectrum and energy efficiency of the forthcoming and future wireless networks (5G and beyond), network service providers (SPs) have investigated and deployed several innovation technologies such as massive MIMO and mmWave. However, the required high complexity and high hardware cost are still the main hindrance to their implementation in practice. Recently, intelligent reflecting surface (IRS) has emerged as a new and cost-effective solution for the SPs. An IRS is generally composed of a large number of passive elements, each of which is able to reflect the incident signal with an adjustable phase-shift. By intelligently tuning the phase-shifts of all elements adaptive to dynamic wireless channels, the signals reflected by the IRS can add constructively with non-reflected signals at the user receiver to boost the received signal power and enhance the data throughput at the user receiver. As such, IRS allows the SPs to improve the spectrum and energy efficiency, extend the network coverage and enhance the Quality of Serivce (QoS) of the users with low cost.


Apart from the above benefits, IRS enables the SPs to provide a new network resource and a new service to the mobile users. In particular, the SPs can provide IRS resources to the mobile users in addition to traditional network resources such as antenna, spectrum, and power. Indeed, there are some works that have investigated the IRS resource allocation issues. In particular, the authors in~\cite{gao2020reflection} consider the allocation of transmit power and IRS resources to mobile users. The network includes one base station (BS) and one IRS. The IRS is divided into modules of reflection elements, i.e., reflection modules. Then, the problem is to determine the number of reflection modules, the corresponding passive beamforming, and transmit power for the users to maximize the signal-to-interference-plus-noise ratio (SINR). Note that the authors consider the allocation of reflection modules, i.e., instead of all the reflection elements, to the users since triggering all the reflection elements frequently can result in the increased latency of adjusting phase-shift. To solve this problem, the parallel alternating direction method of multipliers (PADMM) algorithm is used. Different from~\cite{gao2020reflection}, the authors in~\cite{gao2020stackelberg} assume that the BS and the IRS belong to different SPs. To maximize the individual utility of the BS and IRS, the Stackelberg game is proposed to jointly optimize the IRS resource price, the transmit power, and the passive beamforming of the triggered reflection modules. However, both the works in \cite{gao2020reflection} and~\cite{gao2020stackelberg} consider scenarios with a single IRS. 

In this letter, we consider an IRS-assisted wireless network with multiple SPs and multiple mobile users. The SPs deploy BSs along with IRSs and provide network services to the users. In particular, the SPs are responsible for allocating the transmit power and IRS resources to the users for their data transmissions. To satisfy different QoS requirements of the mobile users, similar to other works, e.g., \cite{gao2020reflection},~\cite{gao2020stackelberg}, we assume that each SP divides the transmit power and the IRS into different power levels and reflection modules, respectively. The SPs may set different prices for their resources, and the users that select different SPs and services may achieve different utilities. Due to their rationality, the users with low utilities have an incentive to adapt their SP and service selections. In other words, the users can dynamically change their SP and service selection strategies over time. To model the dynamic SP and service selection strategies of the users, we propose to use the evolutionary game~\cite{hofbauer2003evolutionary}. The reason for the use of the evolutionary game is that this game is able to deal with the problem of dynamic selection strategies of players, i.e., the users.  In particular, the players in the evolutionary game are bounded rational, and they can adapt their strategies gradually to reach the evolutionary equilibrium. Furthermore, the algorithm to implement the strategy adaptation based on the evolutionary game has a low complexity that is suitable for the dynamic strategy selections of the users~\cite{gao2019dynamic}.

Our main contribution is as follows: we first formulate the SP and service selection in the network as an evolutionary game. In the game, the users form populations, and they adjust their SP and service selections based on their utilities. We then model the SP and service adaptation of the users as replicator dynamics in the evolutionary game and analyze the equilibrium of the evolutionary game. Finally, we provide performance evaluation to demonstrate the consistency with the analytical results and to validate the proposed game model. 

\section{System Model}
\begin{figure}[h]
 \centering
\includegraphics[width=0.85\linewidth]{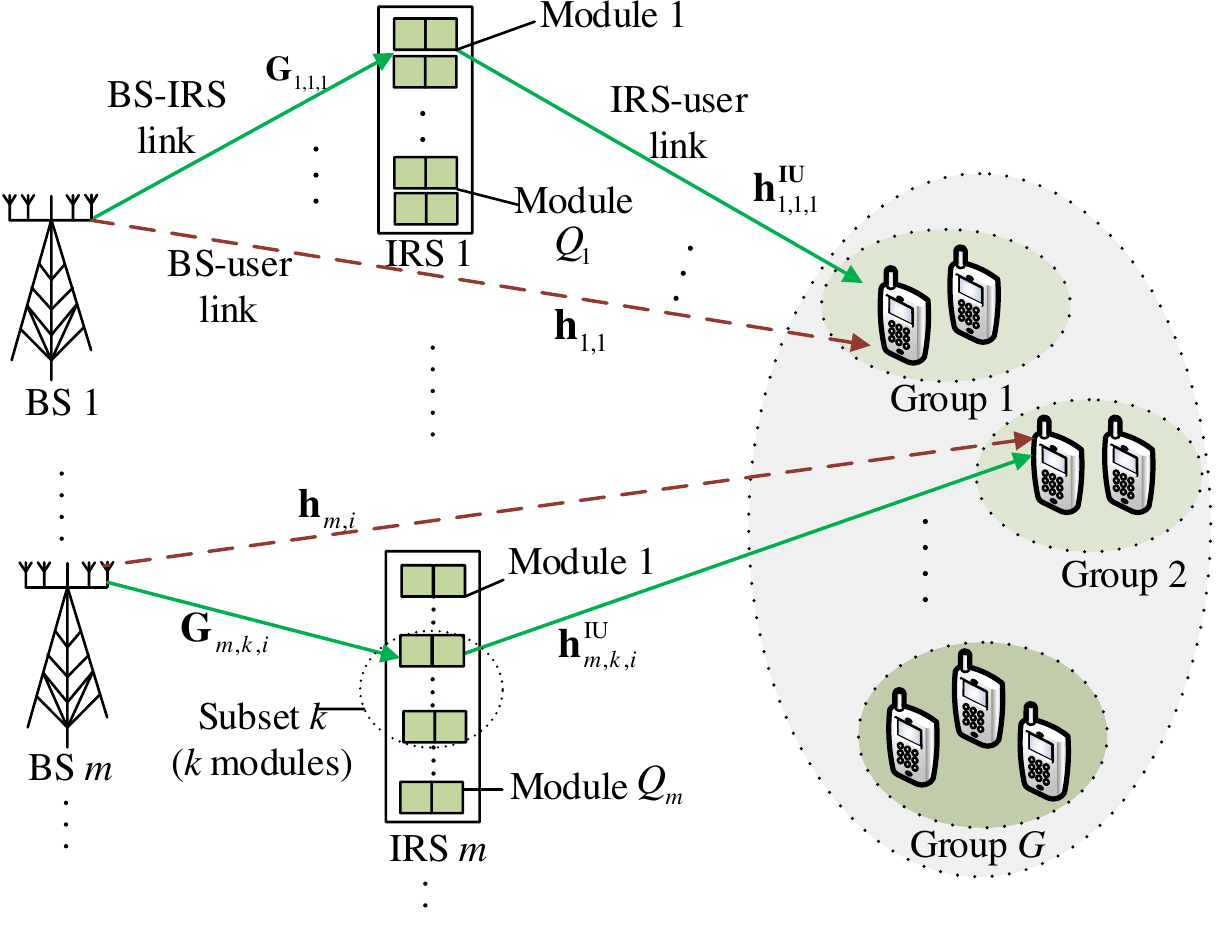}
 \caption{A wireless network with multiple IRS modules.}
  \label{IRS_model_selection}
\end{figure}
We consider a system model as shown in~Fig.~\ref{IRS_model_selection} that consists of a set $\mathcal{M}$ of $M$ BSs and a set $\mathcal{N}$ of $N$ single-antenna users. The $M$ BSs are assumed to belong to $M$ service providers (SPs), and each BS $m$ is equipped with $L_m$ antennas. Note that our model can be extended to a general case in which one SP deploys multiple BSs and IRSs. To avoid the intra-cell interference, each SP uses time-division multiplex access (TDMA) to deploy the IRS-enhanced communication service. Also, frequency-division multiplex access (FDMA) is used to avoid the inter-cell interference among the users belonging to different BSs. Let $B_m$ denote the bandwidth assigned to BS $m$. To provide flexible services to the users, BS $m$ has a set $\mathcal{P}_m$ of $P_m$ power levels, denoted by $\{J_{m,1},\ldots,J_{m,P_m}\}$, that the users can select for their transmissions. We assume that $J_{m,1} <J_{m,2}< \dots < J_{m,P_m}$, where $ J_{m,P_m}$ is the maximum power of BS $m$. SP $m$ deploys an IRS, i.e., denoted by IRS $m$, to improve the QoS for the users that BS $m$ serves. IRS $m$ has $K_m$ reflection elements, and the IRS is divided into $Q_m$ modules controlled by parallel switches. Each module in IRS $m$ consists of $E_m$ elements, and thus $K_m = Q_mE_m$. One BS-IRS pair can serve multiple users, but the user is associated with one BS-IRS pair. Moreover, the user is allowed to select one or multiple modules, i.e., a subset of modules, of the selected IRS. In general, given a selected power level and bandwidth $B_m$, the data throughput achieved by the user depends on the number of modules, i.e., instead of the orders/indexes of the modules, in the selected subset. Thus, SP $m$ has $Q_m$ potential subsets of modules that the user can select, and subset $k, 1 \leq k \leq Q_m$, has $k$ modules. Denote $\boldsymbol{\Theta}_{m,k}$ as the phase-shift matrix corresponding to the subset that the user selects, i.e., subset $k$ of IRS $m$. Then, $\boldsymbol{\Theta}_{m,k}$ is a diagonal matrix in which its main diagonal consists of phase-shifts of $kE_m$ reflection elements of IRS $m$. In particular, we have $\boldsymbol{\Theta}_{m,k}=\text{diag}(\theta_{m,k,1},\ldots,\theta_{m,k, kE_m})$, where $\theta_{m,k,e}$ is the phase-shift of reflection element $e$ of subset $k$ of IRS $m$, $\theta_{m,k,e}=e^{j\alpha_{m,k,e}}, \alpha_{m,k,e} \in [0, 2\pi), e=\{1,\ldots,kE_m\}$. With the assistance of subset $k$ of IRS $m$, the signal received at each user $i$ is the sum of 1) the received signal via the direct link and 2) the received signal via the IRS-assisted link. Thus, the received signal at user $i \in \mathcal{N}$ when selecting subset $k$ of IRS $m$ and power level $J_{m,j}$ of BS $m$ is determined as follows:
\begin{equation}
y_{i}=\big{(}\mathbf{h}^{\text{H}}_{m,i} +   (\mathbf{h}^{\rm{IU}}_{m,k,i})^{\text{H}}\boldsymbol{\Theta}^{\text{H}}_{m,k}\mathbf{G}_{m,k}\big{)}\mathbf{w}_{m,j,i}s_{i} + \omega_i,
\label{received_signal_user}
\end{equation}
where $s_{i}$ is the data symbol intended to user $i$, $\mathbf{w}_{m,j,i}  \in \mathbb{C}^{L_m\times1}$ is the beamforming vector associated with $s_{i}$ containing power level $J_{m,j}$ that the user selects, $\mathbf{h}_{m,i} \in \mathbb{C}^{L_m\times1}$ is the vector of channels of the direct link from BS $m$ to user $i$, $\mathbf{h}^{\rm{IU}}_{m,k,i} \in \mathbb{C}^{K_m\times1}$ is the vector of channels of the link from subset $k$ of IRS $m$ to user $i$, $\mathbf{G}_{m,k} \in \mathbb{C}^{K_m \times L_m}$ is the vector of channels from BS $m$ to subset $k$ of IRS $m$, and $\omega_i$ is the complex additive white Gaussian noise at user $i$, $\omega_i \sim \mathcal{CN}(0,\sigma_0^2)$, where $\sigma_0^2$ is the variance. We assume that the channel state information (CSI) of all channels involved is perfectly known at BS $m$, i.e., based on the pilot signals. In addition, IRSs are typically deployed in static environments due to the challenging task of CSI estimation, we can also assume that the quasi-static flat-fading model or even static flat-fading model is applied for all channels~\cite{wu2019intelligent}. The signal-to-noise ratio (SNR) of user $i$ is defined as the ratio of the power of signal targeted to user $i$ and the noise over bandwidth $B_m$ as follows:
\begin{equation}
\eta_{m,k,i} = \frac{|(\mathbf{h}^{\text{H}}_{m,i} + (\mathbf{h}^{\rm{IU}}_{m,k,i})^{\text{H}}\boldsymbol{\Theta}^{\text{H}}_{m,k}\mathbf{G}_{m,k,i})\mathbf{w}_{m,j,i}|^2} {B_m\sigma_0^2}.
\label{SINR_user_i}
\end{equation}

\section{Game Formulation and Equilibrium Analysis}
There are totally $N$ users, $M$ BSs, and $M$ IRSs in the network. SP $m$ offers a set $\mathcal{P}_m$ of $P_m$ power levels and a set $\mathcal{Q}_m$ of $Q_m$ subsets of IRS modules that the user can select. In particular, one subset consists of one or multiple modules of the IRS. Again, one module of IRS $m$ consists of $E_m$ elements. Without loss of generality, we assume that $N$ users are divided into $G=\sum_{m=1}^MP_mQ_m$ groups. Each group, say group $g,1 \leq g \leq G$, consists of $N_{m,k,j}$ users that select IRS $m$, subset $k$, i.e., the corresponding phase-shift matrix $\boldsymbol{\Theta}_{m,k}$, and power level $J_{m,j}$. We can say that users in group $g$ select network service $g$. Note that when the user selects IRS $m$, it is only allowed to select subsets of modules of IRS $m$ and power levels of BS $m$. We have $\sum_{m=1}^M\sum_{k=1}^{Q_m}\sum_{j=1}^{P_m}N_{m, k,j}=N$, and each user in group $g$ selects IRS $m$, $\boldsymbol{\Theta}_{m,k}$, and $J_{m,j}$ at a probability of $p_{m,k,j}=N_{m,k,j}/N$. The phase-shift matrix of the IRS is typically optimized for a group of nearby users, and we assume that the channel statistics are identical for those users in the same group~\cite{gao2019dynamic}. Therefore, the users in the group can be represented by an index set $(m,k,j)$, and the expected downlink transmission rate of the user in the group is 
 \begin{equation}
 \overline{R}_{m,k,j}=\frac{B_m}{p_{m,k,j}N}
 \log_2\big{(}1+\eta_{m,k,j}\big{)},
 \end{equation}
 where $\eta_{m,k,j}=\frac{|(\mathbf{h}^{\text{H}}_{m,j} +(\mathbf{h}^{\rm{IU}}_{m,k,j})^{\text{H}}\boldsymbol{\Theta}^{\text{H}}_{m,k}\mathbf{G}_{m,k})\mathbf{w}_{m,j}|^2} {B_m\sigma_0^2}$, where $\mathbf{w}_{m,j}$ refers to the beamforming vector associated with the users in the group that selects power level $J_{m,j}$. Let $v_{m,k,j}$ denote the value of unit data to the user in group $g$ when selecting  IRS $m$, subset $k$, and power level $J_{m,j}$. Denote $\gamma_m^I$ as the price per element in IRS $m$ and $\gamma_m^P$ as the price per unit power. Prices $\gamma_m^I$ and $\gamma_m^P$ are set by SP $m$ that are constant. Since users in each group share the same resources, they should share the resource cost. Then, the utility of the user is given by
\begin{equation}
 u_{m,k,j} = v_{m,k,j} \overline{R}_{m,k,j} - \big{(}\gamma_m^I \norm{\boldsymbol{\Theta}_{m,k}}_0 - \gamma_m^P J_{m,j}\big{)}/p_{m,k,j}N,
 \label{utility_user}
\end{equation}
where the $l_0$-norm is used to count the number of non-zero elements of a diagonal matrix that here refers to the number of active reflection elements of IRS $m$ that the user selects. The average utility of the user is $\overline{u}_{m,k,j}=\sum_{m=1}^M\sum_{k=1}^{Q_m}\sum_{j=1}^{P_m} p_{m, k,j} u_{m, k,j}$. We leverage the replicator dynamics to model the SP and service adaptation of the users. The replicator dynamic process of the users is expressed as a series of ordinary differential equations as follows:
\begin{align}
\notag
\dot{p}_{m, k,j}(t) =&\mu p_{m, k,j}(t)[u_{m, k,j}(t)-\overline{u}(t)], \\
& m \in \mathcal{M}, k \in \mathcal{Q}_m, j \in \mathcal{P}_m,  \forall t,
\label{replicator_evolu}
\end{align}
where $\dot{p}_{m, k,j}(t)$ represents the first derivative of $p_{m, k,j}$ with respect to $t$, and $p_{m, k,j}(t_0)=p^0_{m, k,j}$ is the initial strategy of the users in group $g$ at $t_0$. The factor $\mu$ is the learning rate of the users that evaluates the strategy adaptation frequency.

We analyze the equilibrium of the evolutionary game defined in (\ref{replicator_evolu}). For presentation simplification, let $g, 1 \leq g \leq G,$ denote the combination of index $(m,k,j)$. Let $f_g(t,p_g)=\mu p_g(t)[u_g(t)-\overline{u}(t)]$. Equation (\ref{replicator_evolu}) can be re-written as follows:
\begin{equation}
\dot{p}_g(t)  =f_g(t,p_g), p_g(t_0)=p^0_g, g =(1,\ldots,G).
\label{replicator_evolu_2}
\end{equation}

\begin{theorem}~\cite{picard_theorem}
Suppose that $f_g(t,p_g)$ and $\frac{\partial f_g}{\partial p_g}(t,p_g)$ are continuous functions in some open rectangle $\{ (t,p_g): 0 \leq  t \leq  \tau, 0 < p_g \leq 1\}$ that contains the point $(t_0,p^0_g)$. Then, the problem (\ref{replicator_evolu_2}) has a unique solution in the closed interval of $I=[t_0-h, t_0+h]$, where $h> 0$. Moreover, the Picard iteration defined by 
\begin{equation}
p^{l+1}_g(t)  =p^0_g + \int_{t_0}^t f_g(t,p^{l}_g(t) ) \;\mathrm{d}t
\label{Picard_iteration}
\end{equation}
produces a sequence of functions ${ p^{l}_g(t)}$ that converges to the solution uniformly on $I$.
\label{theorem_picard}
\end{theorem}

\begin{proof}
Please refer to~\cite{picard_theorem} for the detailed proof.
\end{proof}

Here, we prove that $f_g(t,p_g)$ and $\frac{\partial f_g}{\partial p_g}(t,p_g)$ are continuous functions in the rectange $\{ (t,p_g): 0 \leq  t \leq  \tau, 0\leq p_g \leq  1\}$. Indeed, it is clear that function $p_g(t)=\frac{N_g(t)}{N}$ is continuous at every $t_0 \in [ 0, \tau]$. Moreover, due to the static flat-fading channel model, the channel vectors $\mathbf{h}_{m,j}(t)$, $\mathbf{h}_g^{\rm{IU}}(t)$, and $\mathbf{G}_g(t)$ are constant and thereby continuous at every $t_0 \in [ 0, \tau]$. Therefore, $\eta_{g}(t)$ is continuous at every $t_0 \in [ 0, \tau]$, and functions $\overline{R}_g(t), u_g(t)$, and $\overline{u}(t)$ are also continuous at every $t_0 \in [ 0, \tau]$ if $p_g(t_0)\neq 0$. Since $f_g(t,p_g)=\mu p_g(t)[u_g(t)-\overline{u}(t)]$ and $\frac{\partial f_g}{\partial p_g}(t,p_g)=\mu[u_g(t)-\overline{u}(t)]$, then $f_g(t,p_g)$ and $\frac{\partial f_g}{\partial p_g}(t,p_g)$ are continuous functions in the open rectangle $\{ (t,p_g): 0 \leq  t \leq  \tau, 0 < p_g \leq 1\}$. According to Theorem~(\ref{Picard_iteration}), problems in (\ref{replicator_evolu_2}) and (\ref{replicator_evolu}) converge to a unique solution. This is verified by simulations in the next section. 

Note that to make the decision on SP and service selections, the users need information about the average utility, i.e., $\overline{u}_{m,k,j}$, and proportion of users choosing different strategies, i.e., $p_{m,k,j}$, from the BSs. However, the up-to-date information may not be available at the users due to the communication latency. Therefore, at time $t$, the users base on the information at time $t-\delta$, i.e., delay for $\delta$ units of time, to make the SP and service selections. Thus, the delayed replicator dynamic process is 
\begin{align}
\notag
\dot{p}_{m, k,j}(t) =&\mu p_{m, k,j}(t-\delta)[u_{m, k,j}(t-\delta)-\overline{u}(t-\delta)], \\
\notag
& m \in \mathcal{M}, k \in \mathcal{Q}_m, j \in \mathcal{P}_m,  \forall t.
\label{replicator_evolu}
\end{align}

Note that as delay $\delta$ is large, the decisions of the users based on the outdated information tend to be inaccurate. As a result, the evolutionary game with the delayed replicator dynamics may not converge. How to determine $\delta^*$ such that the evolutionary game converges is challenging. As an example, consider a simple scenario with $M=2$, and each SP $m$ offers one service including subset $\boldsymbol{\Theta}_{m}$ and power level $J_m$:
\begin{theorem}
The stability of the evolutionary equilibrium with the delayed replicator dynamics can be guaranteed if and only if
\begin{equation}
\delta^* < \frac{\pi}{2\mu \sum_{m\in \mathcal{M}} \frac{B_m\log_2(1 + \eta_m) -\gamma_m^I||\boldsymbol{\Theta}_{m}||_0-\gamma_m^PJ_m }{N} }
\label{stability}
\end{equation}
\label{theorem_stability}
\end{theorem}
\begin{proof}
The detailed proof can be derived by following~\cite{gao2019dynamic} which is omitted from here.  
\end{proof}

Theorem (\ref{theorem_stability}) shows that the stability of the evolutionary equilibrium is guaranteed as the users use the information at $t< \delta^*$ for their decisions.

\section{Performance Evaluation}
In this section, we present the numerical results to demonstrate the effectiveness of the proposed dynamic SP and service selection in the IRS-assisted wireless network. We consider a network with $2$ SPs, $2$ BSs, $2$ IRSs and $100$ users. The sizes of IRS 1 and IRS 2 are $8$ elements. SP 1 divides IRS 1 to $2$ modules, and SP 2 does not divide IRS 2. Each BS is equipped with $4$ antennas and offers $2$ power levels that the users can select: $J_{1,1}=15$ dBm, $J_{1,2}=30$ dBm, $J_{2,1}=10$ dBm, and $J_{2,1}=20$ dBm. As such, SP 1 offers 4 services, namely Services 1, 2, 3, and 4, and SP 2 offers 2 services, namely Services 1 and 2. All the channels suffer a Rayleigh fading, and the path loss at the reference distance  $d_0=1$ m is $-30$ dB. Due to the obstacles, the pass loss of the direct links from the BS to the users is much higher than those of the links between the BS and the IRS as well as the links between the IRS and the users, we can thus set $\alpha_{m,k}=\alpha_{m,k,j}=2$, and $\alpha_{m,j}=6$. When the user selects a service of the SP, the BS uses the fixed point iteration algorithm~\cite{yu2019miso} to optimize the beamforming vector and the phase-shift matrix of the IRS subset for the user.
\begin{figure}[]
	\centering
	\begin{minipage}[t]{0.24\textwidth}
	 \includegraphics[width=1\textwidth]{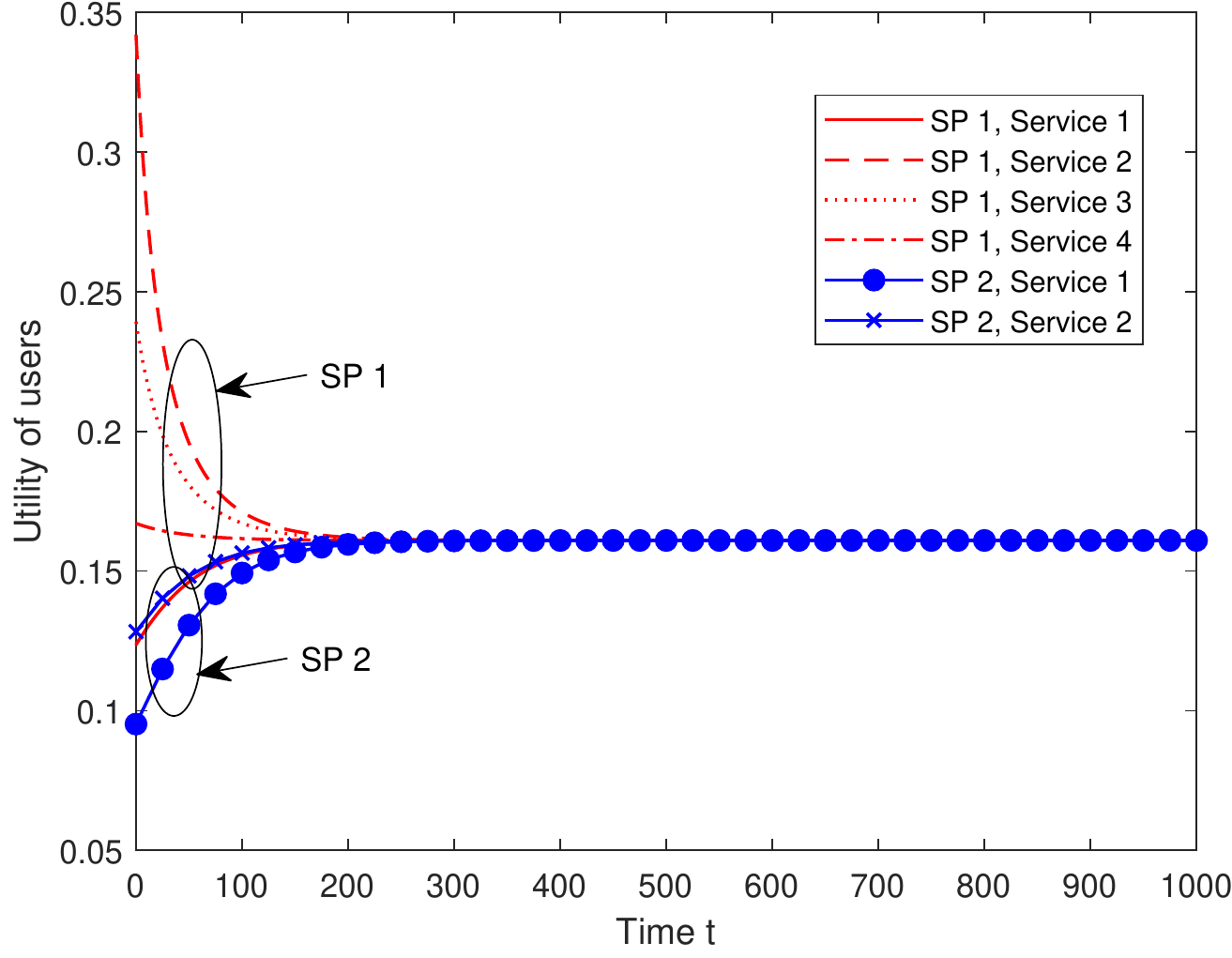}
	\subcaption{}
	\end{minipage}
	\begin{minipage}[t]{0.24\textwidth}
		\includegraphics[width=1\textwidth]{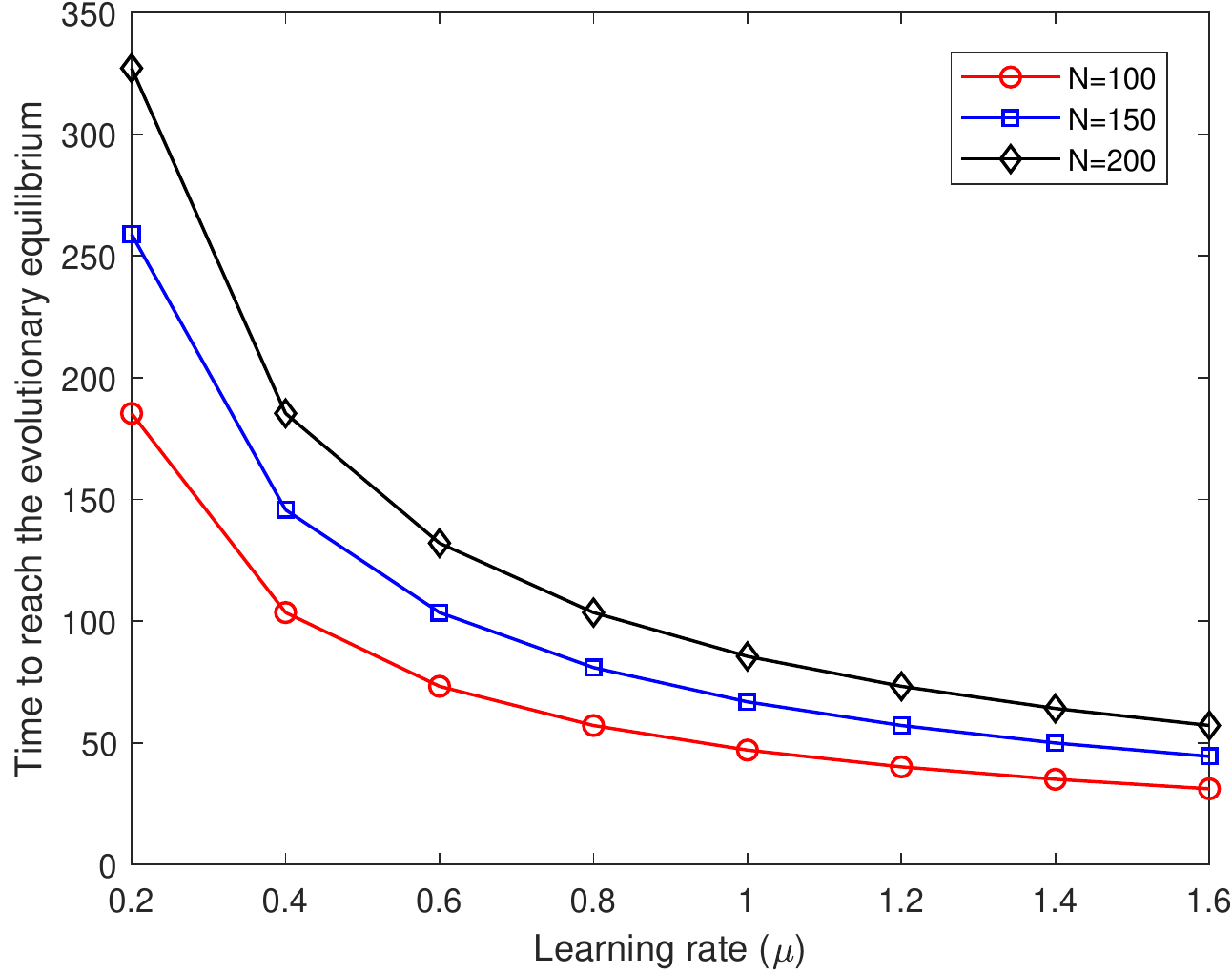}
		\subcaption{}
	\end{minipage}
		\caption{(a) Utility of user groups vs. time and (b) time to reach the equilibrium.}
	\label{equilibrium}
\end{figure}

First, it is important to verify the equilibrium convergence of the game scheme. Figure~\ref{equilibrium}(a) shows the utilities of the users selecting different SPs and services versus time. As seen, the utilities of the users selecting different SPs and services vary until the evolutionary equilibrium is reached. Also, at the evolutionary equilibrium, the users achieve the same utility although they select different SPs and services. The reason is that the evolutionary equilibrium is reached only when the utilities of the users selecting any service provided by any SP are equal to the expected utilities of the users. 
\begin{figure}[]
	\centering
	\begin{minipage}[t]{0.24\textwidth}
		\includegraphics[width=1\textwidth]{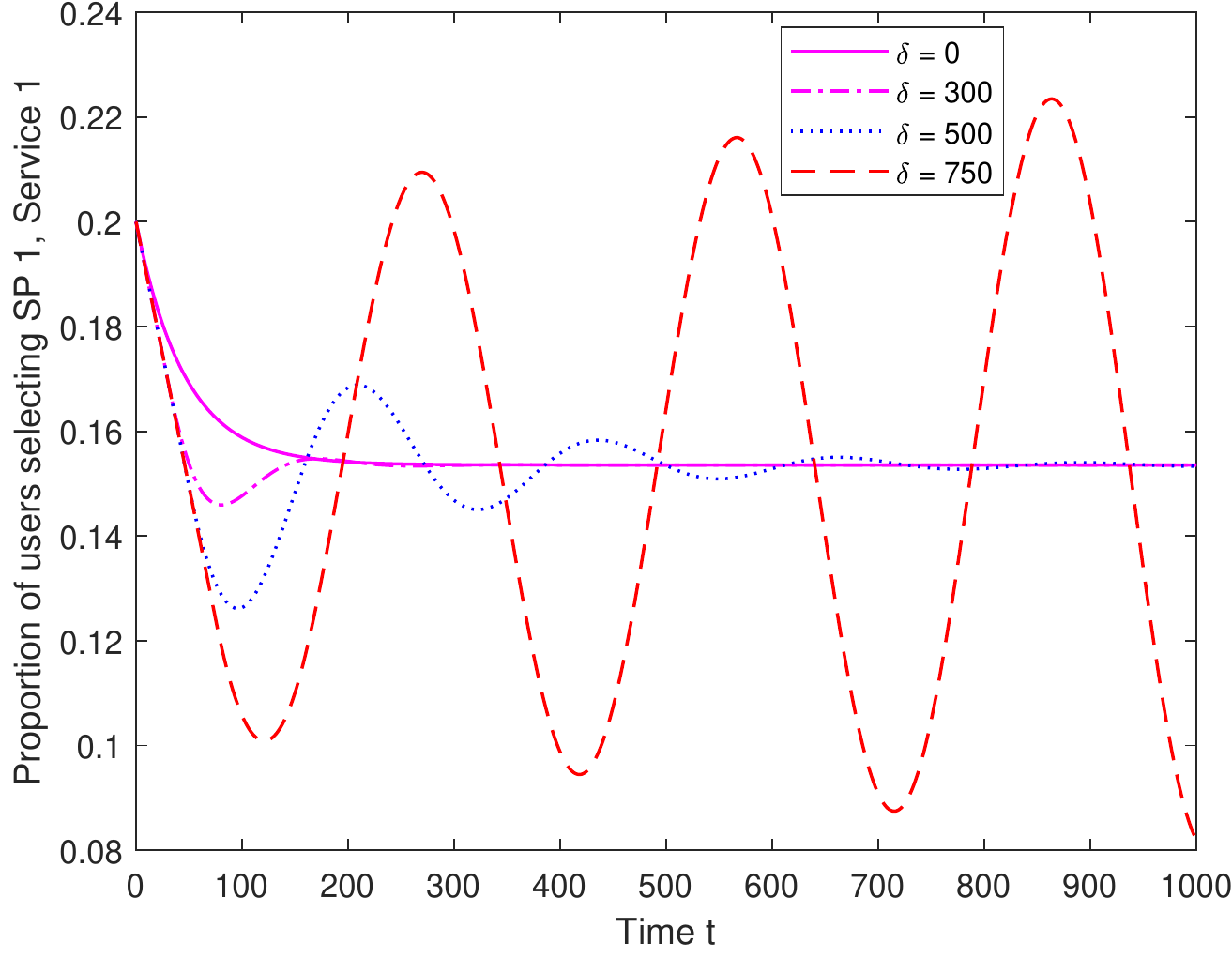}
		\subcaption{}
	\end{minipage}
	\begin{minipage}[t]{0.24\textwidth}
		\includegraphics[width=1\textwidth]{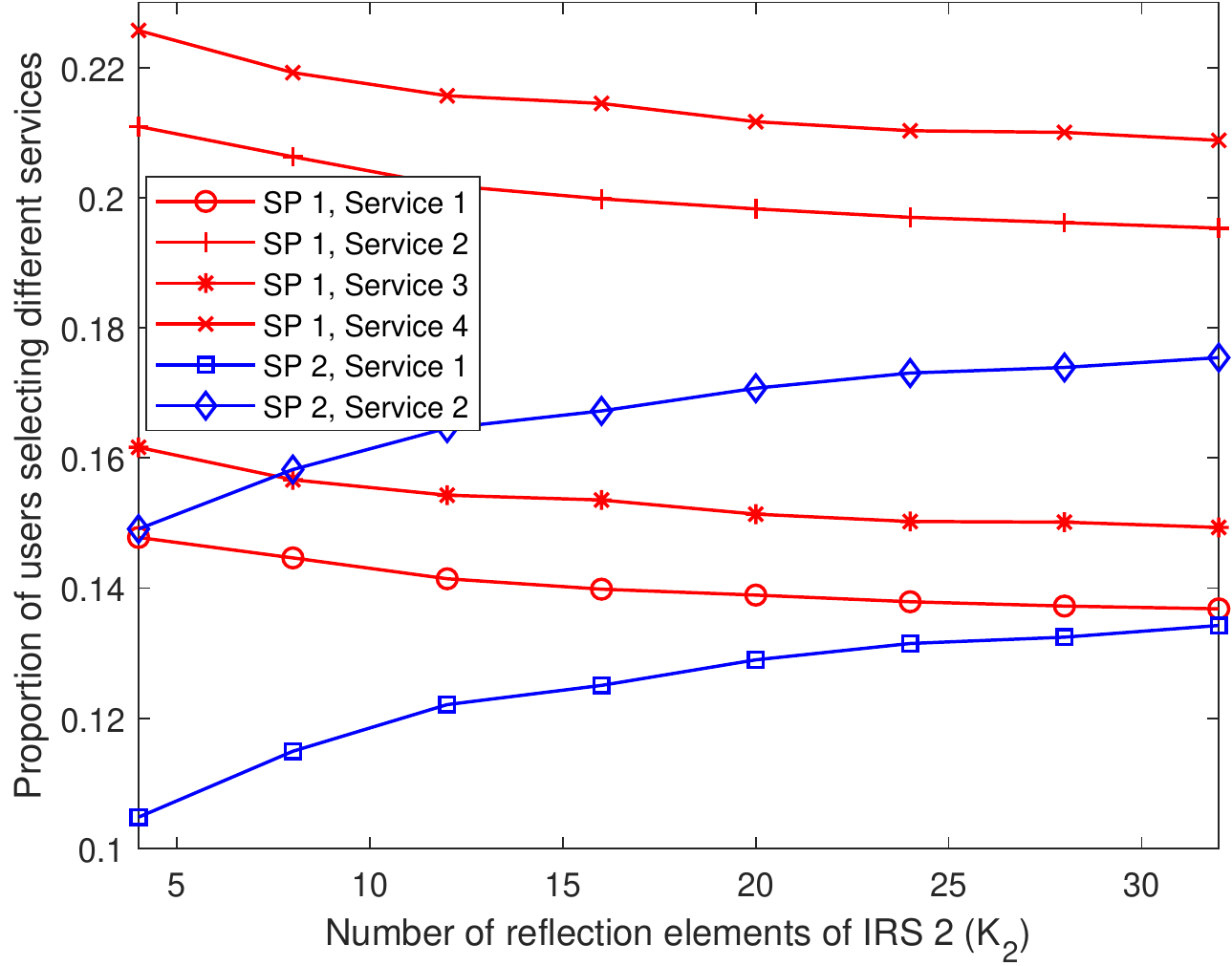}
		\subcaption{}
	\end{minipage}
	\caption{(a) Proportion of the users selecting SP1 and Service 1 with different delays and (b) proportions of the users selecting different SPs and services versus the size of IRS 2.}
	\label{proportion}
\end{figure}

Note that the time to reach the evolutionary equilibrium can be different depending on the learning rate $\mu$ and the number of users $N$ as shown in Fig.~\ref{equilibrium}(b). As seen, the evolutionary equilibrium is reached faster as the value of $\mu$ is higher. The reason is that the frequency of the strategy adaptation of the users is higher with the high value of $\mu$. Moreover, for a given value of $\mu$, more time is required to reach the evolutionary equilibrium as the number of users is higher. 

Next, we discuss the impact of information delay $\delta$ on the proportions of users selecting different SPs and services. For the evaluation purpose, we consider the proportion of the users selecting Service 1 provided by SP 1 as shown in Fig.~\ref{proportion}(a). As seen, when $\delta>0$, there is a fluctuating dynamics of strategy adaptation. The fluctuation becomes larger as $\delta$ increases, and the service selection cannot reach the evolutionary equilibrium as $\delta \geq 750$. This is due to the fact that the outdated information makes the users' decisions accurate. 

Now, we investigate the impact of sizes of IRSs on the proportions of users selecting different SPs and services. In particular, we vary the size $K_2$ of IRS 2 provided by SP 2. Note that $K_2$ is also the size of one module for trading since SP 2 does not divide IRS 2 into different modules. As shown in Fig.~\ref{proportion}(b), as the size of IRS 2 increases, the proportions of users selecting services provided by SP 2 increase. The reason is that as $K_2$ increases, the throughput and utility obtained by the users selecting services provided by SP 2 increase. However, as the size of IRS 2 is large, the increasing rate tends to be slower. This is because of that the users pay a very high resource cost as they select the services provided by SP 2. As a result, their utilities decrease, and they tend to select the services provided by SP 1. 

\begin{figure}[]
\centerline{\includegraphics[width=0.25\textwidth]{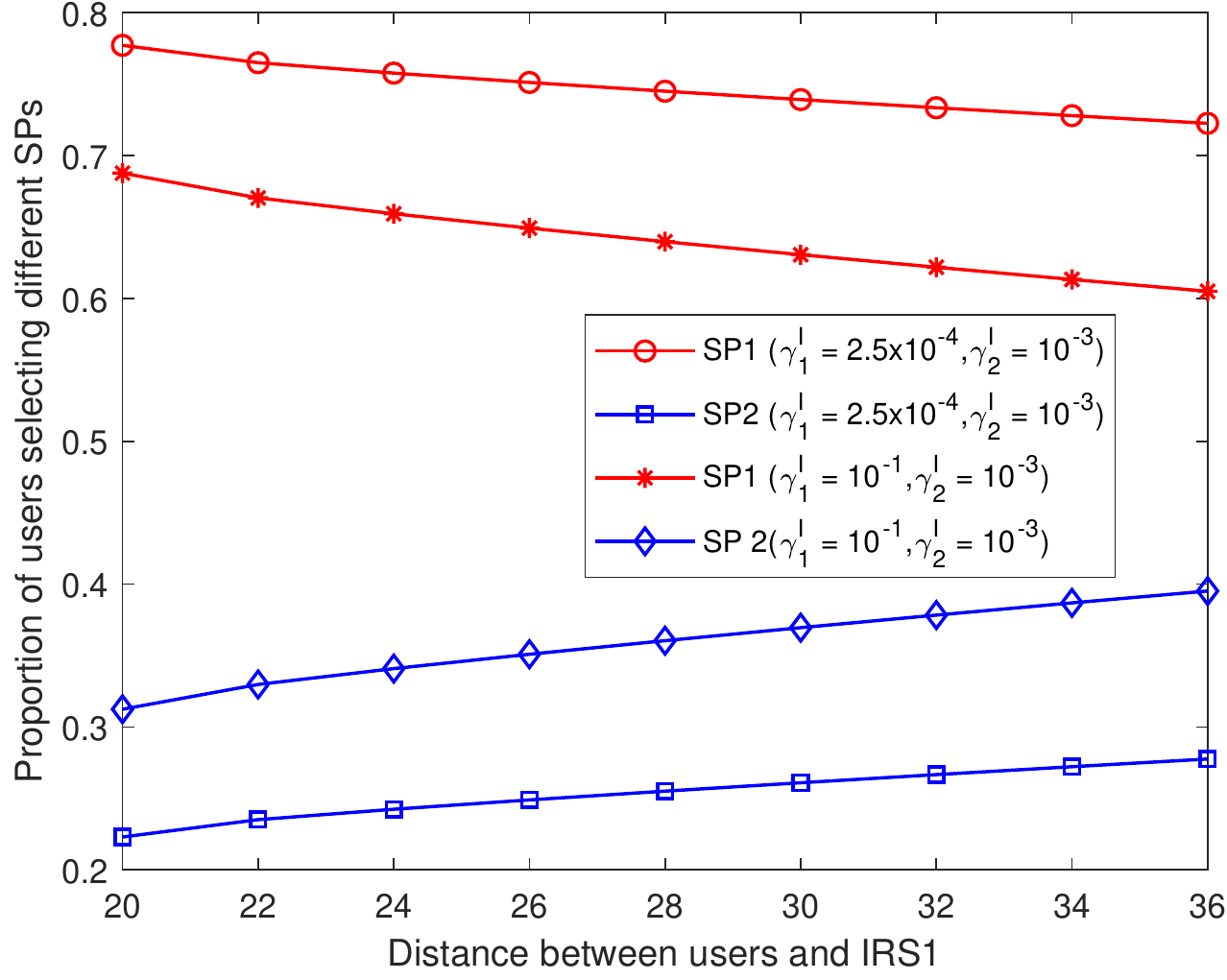}}
	\caption{Proportions of users selecting different SPs vs. distance and price.}
	\label{distance_price}
\end{figure}

We finally discuss how the mobility of the users impacts on the SP and service selection of the users. Figure~\ref{distance_price} shows the proportions of users selecting different SPs as the distance between the users and IRS 1 varies for different IRS prices set by SP 1. As observed, for a given price, the proportion of users selecting SP 2 increases as the distance between the users and IRS 1 increases. The reason is that as the distance between the users and IRS 1 increases, the throughput obtained by the users if selecting services of SP 1 decreases. Thus, the users are willing to select services of SP 2. Note that as SP 1 increases its service price, the utilities of the users currently selecting services of SP 1 decrease, and thus the proportion of users selecting SP 1 decreases as shown in the figure.

\section{Conclusions}
In this letter, we have investigated the dynamic SP and service
selection in an IRS-assisted wireless network. Specifically, we have formulated a joint SP and
service selection problem as an evolutionary game. We have modeled the SP and service adaptation of the users as replicator dynamics and analyzed the equilibrium of the evolutionary game. We have provided performance evaluation to demonstrate the consistency with the analytical results and to validate the proposed game model. 

\bibliographystyle{IEEEtran}
\bibliography{IRS_selection}{}



\end{document}